\documentclass[onecolumn, 12pt]{IEEEtran}

\ifCLASSINFOpdf
\else
\fi

\usepackage{bm}
\usepackage{color}
\usepackage{float}
\usepackage{amsthm}
\usepackage{amsmath}
\usepackage{amssymb}
\usepackage{amsfonts}
\usepackage{fancyhdr}
\usepackage{graphicx}
\usepackage{placeins}
\usepackage{enumerate}
\usepackage{subfigure}
\usepackage{algorithmic}
\usepackage{tabularx,colortbl}
\usepackage{epstopdf}
\usepackage[numbers,sort&compress]{natbib}
\usepackage{setspace}

\newtheoremstyle{mytheorem}
{0pt} 
{0pt} 
{\normalfont} 
{1em} 
{\em} 
{:} 
{0.5em} 
{} 
\theoremstyle{mytheorem}

\newtheorem{mytheorem}{Theorem}
\newtheorem{mylemma}{Lemma}
\newtheorem{myremark}{Remark}

\begin{document}

\title{Ergodic Rate Analysis of Cooperative Ambient Backscatter Communication}

\author{
Shaoqing~Zhou,
Wei~Xu,~\IEEEmembership{Senior Member,~IEEE,}
Kezhi~Wang,~\IEEEmembership{Member,~IEEE,}
Cunhua~Pan,~\IEEEmembership{Member,~IEEE,}
Mohamed-Slim~Alouini,~\IEEEmembership{Fellow,~IEEE,}
and~Arumugam~Nallanathan,~\IEEEmembership{Fellow,~IEEE}

\thanks{S. Zhou and W. Xu are with the National Mobile Communications Research Laboratory, Southeast University, Nanjing 210096, China (e-mail: sq.zhou@seu.edu.cn; wxu@seu.edu.cn).}
\thanks{K. Wang is with the Department of Computer and Information Sciences, Northumbria University, Newcastle upon Tyne NE1 8ST, U.K. (e-mail: kezhi.wang@northumbria.ac.uk).}
\thanks{C. Pan and A. Nallanathan are with the School of Electronic Engineering and Computer Science, Queen Mary University of London, London E1 4NS, U.K. (e-mail: c.pan@qmul.ac.uk; a.nallanathan@qmul.ac.uk).}
\thanks{M.-S. Alouini is with the Division of Computer, Electrical, and Mathematical Science and Engineering, King Abdullah University of Science and Technology, Thuwal 23955-6900, Saudi Arabia (e-mail: slim.alouini@kaust.edu.sa).}
}


\maketitle

\begin{abstract}
Ambient backscatter communication has shown great potential in the development of future wireless networks. It enables a backscatter transmitter (BTx) to send information directly to an adjacent receiver by modulating over ambient radio frequency (RF) carriers. In this paper, we consider a cooperative ambient backscatter communication system where a multi-antenna cooperative receiver separately decodes signals from an RF source and a BTx. Upper bounds of the ergodic rates of both links are derived. The power scaling laws are accordingly characterized for both the primary cellular transmission and the cooperative backscatter. The impact of additional backscatter link is also quantitatively analyzed. Simulation results are provided to verify the derived results.
\end{abstract}


\IEEEpeerreviewmaketitle

\section{Introduction}
\IEEEPARstart{A}{mbient} backscatter technology has attracted increasing attention from both indoor and outdoor future wireless applications due to its ability to realize energy efficient and low-cost communication. In ambient backscatter communication, backscatter transmitter (BTx) harvests energy from ambient radio frequency (RF) signals and modulates the received signals to send information to neighboring receivers \cite{Boyer2014Backscatter}. The RF source can usually be a base station in cellular network while the BTx can be an Internet-of-things (IoT) device with a single-antenna due to power limitations. The cellular receiver and the backscatter receiver may be the same device or separate terminals. The receiver is usually named as cooperative receiver (CRx) in the former case, which could be smart phone, laptop or other cellular devices. A typical application is that a smart phone in a cellular network simultaneously recovers both information from a cellular base station and a wearable sensor of body-area-network application.

One of the main challenges of ambient backscatter is to deal with severe direct-link interference from the RF source. A straightforward solution was to treat the interference as noise and to demodulate the backscattered information via energy detection \cite{Liu2013Ambient}. In \cite{Zhang2016Enabling}, the BTx conducted frequency shifts on backscattered signals to nearby unoccupied frequencies such that direct-link interference was avoided. The operation of frequency shift, however, imposes higher requirements on device hardware. Alternatively, multi-antenna receiver with maximum likelihood detectors and successive interference cancellation (SIC)-based detectors was considered in \cite{Yang2018Cooperative} for cooperative ambient backscatter systems.

Meanwhile, optimized ambient backscatter communication technologies have also been studied. A novel spectrum sharing model was proposed in \cite{Kang2018Riding} considering that the BTx employed the same frequency band as the RF source for transmission. The ergodic capacity of the backscatter link was maximized through simultaneously optimizing the transmit power of RF source and the reflection coefficient of BTx. In \cite{Long2018Transmit}, the transmit beamforming was optimized for rate maximization of a cooperative ambient backscatter communication system with a multi-antenna RF source.

In this paper, we study the ergodic rate performance of a cooperative ambient backscatter communication system. We derive the theoretical expressions for the ergodic rates of the conventional cellular data link and the cooperative backscatter link in the system. In particular, the power scaling laws of both transmitters are discovered and the characteristics of the backscatter link are analyzed with numerical verifications.

\section{System Model}

\begin{figure}
  \setlength{\abovecaptionskip}{0.cm}
  \setlength{\belowcaptionskip}{-0.cm}
  \centering
  \includegraphics[width = 3.5in]{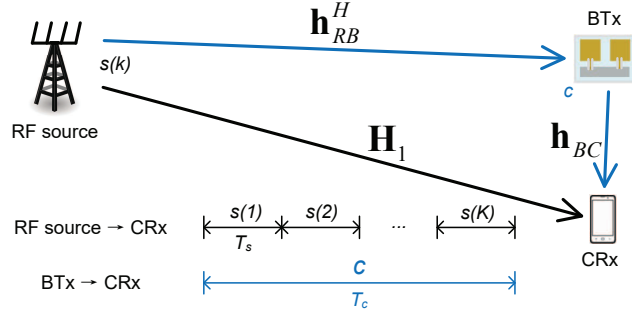}
  \caption{System model.}
  \label{Fig:system_model}
\end{figure}

We consider a cooperative ambient backscatter communication system which consists of an RF source, a single-antenna BTx and a CRx, as shown in Fig. \ref{Fig:system_model}. The RF source, which can be a regular transmitter in cellular network, is equipped with $M$ antennas. It sends information to the CRx and simultaneously provides RF signal waveforms to the BTx for backscattering. The BTx then reflects the modulated RF signals to the CRx. The CRx, which can be a typical cellular terminal \cite{Yang2018Cooperative}, is equipped with $N$ antennas. It receives both the information directly from the RF source and the information from the BTx through backscattering.

From the perspective of spectrum reusing, the transmission from the RF source to the CRx is regarded as the primary link, while the transmission from the BTx to the CRx through backscattering is considered as the secondary link. The CRx needs to retrieve independent information from both transmitters. Considering block-fading channels, the primary link channel, the channel from the RF source to the BTx, and the backscatter channel in each block are respectively denoted by $\mathbf{H}_1 \in \mathbb{C}^{N \times M}$, $\mathbf{h}_{RB}^H \in \mathbb{C}^{1 \times M}$, and $\mathbf{h}_{BC} \in \mathbb{C}^{N \times 1}$. Assume that the RF source knows only $\mathbf{H}_1$ and the CRx has both $\mathbf{H}_1$ and $\mathbf{h}_{BC}$.

Let $s(k)$ be the transmitted signal of the RF source and $\mathbb{E} \left[ \left| s(k) \right| ^2 \right] = P$, where $P$ is the transmitted power budget. Since the BTx usually has much lower rate transmission task than the RF source \cite{Liu2013Ambient}, without loss of generality, we assume that the symbol period of the BTx is $K$ times that of the RF source, as depicted in Fig. \ref{Fig:system_model}. Denote $c$ as the transmitted signal of BTx during $K$ RF source symbol periods. The signal $c$ is assumed to be a random variable with zero mean and unit variance. Then the backscattered signal from the BTx in the $k$th RF source symbol period can be expressed as $ \alpha c \mathbf{h}_{RB}^H \mathbf{w} s(k)$ for $k = 1, 2, \dots, K$, where $\alpha \in (0,1]$ represents the reflection coefficient of BTx and $\mathbf{w} \in \mathbb{C}^{M \times 1}$ is the beamforming vector at the RF source. The CRx receives signals from both transmitters, which is
\begin{equation} \label{y}
\mathbf{y}(k) = \mathbf{H}_1 \mathbf{w} s(k) + \alpha c \mathbf{h}_{BC} \mathbf{h}_{RB}^H \mathbf{w} s(k) + \mathbf{n}(k),
\end{equation}
where $\mathbf{n}(k)$ is the additive Gaussian noise vector with zero mean vector and covariance matrix $\sigma^2 \mathbf{I}$, where $\mathbf{I}$ is an identity matrix.

The secondary link usually experiences more attenuation than the primary link. Thus, based on the decoding strategy of SIC \cite{Long2018Transmit}, the CRx first decodes $s(k)$ and then it removes $s(k)$ from the received signal before detecting $c$. For the $k$th RF source symbol period, the detected signal is written as
\begin{equation}\label{primary_link}
y(k) = \mathbf{v}_s^H \mathbf{y}(k) = \mathbf{v}_s^H (\mathbf{H}_1 + \alpha c \mathbf{h}_{BC} \mathbf{h}_{RB}^H) \mathbf{w} s(k) + \mathbf{v}_s^H \mathbf{n}(k),
\end{equation}
where $\mathbf{v}_s \in \mathbb{C}^{N \times 1}$ is the combining vector of $s(k)$ at CRx. For detection, the transmitted signal $s(k)$ can be regarded passing through an equivalent channel $\overline{\mathbf{H}} \triangleq \mathbf{H}_1 + \alpha c \mathbf{h}_{BC} \mathbf{h}_{RB}^H$. Since the CRx does not know the equivalent channel information due to the unknown $c$, noncoherent detection is adopted for detecting $s(k)$. Given $c$, we obtain the signal-to-noise-ratio (SNR) of $s(k)$ as
\begin{equation} \label{SNR1}
\text{SNR}_{1|c} = \frac{P \vert \mathbf{v}_s^H \left( \mathbf{H}_1 + \alpha c \mathbf{h}_{BC} \mathbf{h}_{RB}^H \right) \mathbf{w} \vert^2}{\sigma^2}.
\end{equation}
Note that, without loss of generality, we assume in (\ref{SNR1}) that the combining vector is normalized, i.e., $\Vert \mathbf{v}_s \Vert^2 = 1$. The capacity of noncoherent detection is consistent with that of coherent detection in the case when the channel is block-fading and the transmission length is large enough \cite{Yingbin2004Capacity}. Correspondingly, for a slowly-varying channel and a sufficiently large $K$, we can write the ergodic rate of the primary link as\footnote{In this paper, we represent $T$ in \cite{Yingbin2004Capacity} as the period of $c$, $T_c$, over the period of $s(k)$, $T_s$, i.e., $K = T_c / T_s$, while Q in \cite{Yingbin2004Capacity} corresponds to the rank of the correlation matrix of the equivalent channel $\overline{\mathbf{H}}$. Moreover, we have the coherence time, $T_{coh}$, satisfying $K < T_{coh} / T_s$ according to [7, \uppercase\expandafter{\romannumeral2}.C].}
\begin{equation} \label{R1}
R_1 = \mathbb{E}_{c,h} \left[ \log_2 \left( 1 + \text{SNR}_{1|c} \right) \right],
\end{equation}
where $\mathbb{E}_{c,h} \left[ \cdot \right]$ denotes the expectation over the BTx symbol and the channel fading. Based on the assumption of channel state information (CSI) and the consideration that the RF source is responsible for the primary link in cellular network, it is natural to choose the beamforming and combining vectors by matching the primary channel, i.e.,
\begin{equation} \label{uandvs}
\mathbf{v}_s = \mathbf{u}_{1m}, \ \mathbf{w} = \mathbf{v}_{1m},
\end{equation}
where $\mathbf{u}_{1m}$ and $\mathbf{v}_{1m}$ respectively represent the corresponding left and right singular vectors of the largest singular value of $\mathbf{H}_1$.

After the RF source signal is detected and then removed from $\mathbf{y}(k)$, the received signal of the secondary link becomes
\begin{equation}
\hat y(k) = \alpha \mathbf{v}_c^H \mathbf{h}_{BC} \mathbf{h}_{RB}^H \mathbf{w} s(k) c + \mathbf{v}_c^H \mathbf{n}(k),
\end{equation}
where $\mathbf{v}_c \in \mathbb{C}^{N \times 1}$ is the combining vector for detecting $c$. During one BTx symbol period, denote $\mathbf{s} = \left[ s(1), s(2), \dots, s(K) \right] ^T$ and $\hat{ \mathbf{y} } = \left[ \hat y(1), \hat y(2), \dots, \hat y(K) \right] ^T$. The received vector of the secondary link can be written as
\begin{equation}
\hat{ \mathbf{y} } = \alpha \mathbf{v}_c^H \mathbf{h}_{BC} \mathbf{h}_{RB}^H \mathbf{w} \mathbf{s} c + \mathbf{n}^{'},
\end{equation}
where $\mathbf{n}^{'} = \left[ \mathbf{v}_c^H\mathbf{n}(1), \mathbf{v}_c^H\mathbf{n}(2), \dots, \mathbf{v}_c^H\mathbf{n}(K) \right] ^T$. Applying the maximal ratio combining, the SNR of detecting $c$ is given as
\begin{equation} \label{SNR2}
\text{SNR}_{2} = K \frac{P \alpha^2 \vert \mathbf{v}_c^H \mathbf{h}_{BC} \mathbf{h}_{RB}^H \mathbf{w} \vert^2}{\Vert \mathbf{v}_c \Vert^2 \sigma^2}.
\end{equation}
Then, the ergodic rate of the secondary link is
\begin{equation} \label{R2}
R_2 = \mathbb{E}_h \left[ \frac{1}{K} \log_2 \left( 1 + \text{SNR}_{2} \right) \right].
\end{equation}
Since the transmit beamforming vector has been determined, the combining vector of the backscattered signal is chosen to match the secondary link channel, i.e., $\mathbf{v}_c = \widetilde{\mathbf{h}}_{BC} \triangleq \mathbf{h}_{BC} / \Vert \mathbf{h}_{BC} \Vert$.

\section{Rate Region Analysis}
In this section, we characterize the rate region of both ergodic rates of the primary and backscatter links and discover the power scaling laws in the cooperative ambient backscatter system.

Before presenting our main results of the derived rate bounds, we first give the useful preliminary results in the following \emph{Lemma 1}.

\subsection{Preliminary Calculations for Rate Analysis}
\begin{mylemma}
Assume that all the channels are independent Rayleigh fading, e.g., $\mathbf{H}_1 \sim \mathcal{CN} (\mathbf{0}_N, \sigma_{1}^2 \mathbf{I}_N)$, $\mathbf{h}_{RB} \sim \mathcal{CN} (\mathbf{0}_M, \sigma_{RB}^2 \mathbf{I}_M)$ and $\mathbf{h}_{BC} \sim \mathcal{CN} (\mathbf{0}_N, \sigma_{BC}^2 \mathbf{I}_N)$. The ergodic rates in (\ref{R1}) and (\ref{R2}) are respectively characterized as
\begin{align} \label{R1upper1}
R_1 &\leq \log_2\left( \frac{P \alpha^2 \sigma_{BC}^2 \sigma_{RB}^2}{\beta \sigma^2} \right) + \log_2e\beta^{-1} G\begin{smallmatrix}4\ 1\\2\ 4\end{smallmatrix} \left[\beta^{-1}\left| \begin{smallmatrix}-1&0\\0&0&-1&-1\end{smallmatrix} \right.\right], \\ \label{R2equ}
R_2 &= \frac{\log_2e}{K\Gamma(N)} \gamma^{-\frac{N+1}{2}} G\begin{smallmatrix}4\ 1\\2\ 4\end{smallmatrix} \left[\gamma^{-1}\left| \begin{smallmatrix}-\frac{N+1}{2}&-\frac{N-1}{2}\\\frac{N-1}{2}&-\frac{N-1}{2}&-\frac{N+1}{2}&-\frac{N+1}{2}\end{smallmatrix} \right.\right],
\end{align}
where $\beta = \frac{P\alpha^2 \sigma_{BC}^2 \sigma_{RB}^2}{\sigma^2 + P\left(\sqrt{M} + \sqrt{N}\right)^2}$, $\gamma = \frac{P K \alpha^2 \sigma_{BC}^2 \sigma_{RB}^2}{\sigma^2}$, and $G \begin{smallmatrix}m\ n\\p\ q\end{smallmatrix} \left[z\left| \begin{smallmatrix}(a_p)\\(b_q)\end{smallmatrix} \right.\right]$ is the Meijer's G-function \cite{Adamchik1990The}.
\end{mylemma}

\def\QEDopen{{\setlength{\fboxsep}{0pt}\setlength{\fboxrule}{0.2pt}\fbox{\rule[0pt]{0pt}{1.3ex}\rule[0pt]{1.3ex}{0pt}}}}
\def\QED{\QEDopen}\def\endproof{\hspace*{\fill}\QED\par\endtrivlist\unskip}
\begin{proof}
Start with the proof of the upper bound of $R_1$. By substituting (\ref{SNR1}) and (\ref{uandvs}) into (\ref{R1}), and denoting $\sigma_{1m}$ as the largest singular value of $\mathbf{H}_1$, we have
\begin{equation} \label{R1upper2}
\begin{split}
R_1 
&\overset{(a)} \leq \mathbb{E}_h \left[ \log_2 \left( 1 + \frac{P}{\sigma^2} \mathbb{E}_{c} \left[ \vert\sigma_{1m} + \alpha c \mathbf{u}_{1m}^H \mathbf{h}_{BC} \mathbf{h}_{RB}^H \mathbf{v}_{1m}\vert^2 \right] \right) \right] \\
&\overset{(b)} = \mathbb{E}_h \left[ \log_2 \left( 1 + \frac{P}{\sigma^2} (\sigma_{1m}^2 + \alpha^2 \vert\mathbf{u}_{1m}^H \mathbf{h}_{BC} \vert^2 \vert \mathbf{h}_{RB}^H \mathbf{v}_{1m}\vert^2 ) \right) \right],
\end{split}
\end{equation}
where $(a)$ applies Jensen's Inequality, and $(b)$ follows from the fact that $c$ is a zero-mean and unit-variance variable and independent of the channel parameters.

To evaluate (\ref{R1upper2}), we need the distribution of the product $\vert\mathbf{u}_{1m}^H \mathbf{h}_{BC} \vert^2 \vert \mathbf{h}_{RB}^H \mathbf{v}_{1m}\vert^2$. For notational simplicity, denote $X \triangleq \vert \mathbf{u}_{1m}^H \widetilde{\mathbf{h}}_{BC} \vert^2$, $Y \triangleq \Vert \frac{\sqrt{2}}{\sigma_{BC}} \mathbf{h}_{BC} \Vert^2$, and $Z \triangleq \vert \frac{\sqrt{2}}{\sigma_{BC}} \mathbf{u}_{1m}^H \mathbf{h}_{BC} \vert^2 = XY$. Note that $X$ is Beta distributed with parameters $1$ and $N-1$ since it is the squared absolute inner product of two uniformly distributed normalized vectors \cite{Chun2007On} and $Y$ is distributed as $\chi^2_{2N}$. We have the probability density function (PDF) of $Z \geq 0$ as
\begin{equation} \label{ZPDF}
\begin{split}
f_Z(z) &\overset{(c)} = \int_z^{\infty} f_X\left(\frac{z}{y}\right) f_Y(y) \frac{1}{y} \mathrm{d}y \\
&\overset{(d)} =\frac{N-1}{2^N\Gamma(N)} \int_z^{\infty} (y-z)^{N-2} e^{-y/2} \mathrm{d}y \\
&\overset{(e)} = \frac{(N-1)e^{-z/2}}{2^N\Gamma(N)} \times \frac{(N-2)!}{2^{-N+1}} \\
&= \frac{1}{2} e^{-z/2},
\end{split}
\end{equation}
where $\Gamma(\cdot)$ is the Gamma function, $(c)$ uses the independency of $X$ and $Y$ \cite{Xu2010MIMO} and the fact that $f_X(x) =0$ for $x>1$, $(d)$ is obtained by substituting the PDFs $f_X(x)=(N-1)(1-x)^{N-2},\ 0 \leq x \leq 1$ and $f_Y(y) = \frac{1}{2^N\Gamma(N)}e^{-y/2}y^{N-1},\ y > 0$, and $(e)$ uses the integral in\setcitestyle{open={},close={}} [\cite{Gradshteyn2007Table}, Eq. (3.3513)].

Similarly, we can obtain that the PDF of $\vert \frac{\sqrt{2}}{\sigma_{RB}} \mathbf{h}_{RB}^H \mathbf{v}_{1m}\vert^2$ is the same exponential distribution as (\ref{ZPDF}) with parameter $1/2$. Now let $A \triangleq \vert \frac{\sqrt{2}}{\sigma_{BC}} \mathbf{u}_{1m}^H \mathbf{h}_{BC} \vert^2 \vert \frac{\sqrt{2}}{\sigma_{RB}} \mathbf{h}_{RB}^H \mathbf{v}_{1m}\vert^2$, we have
\begin{equation} \label{PDF}
f_A(a) = \int_0^{\infty} f_Z\left(\frac{a}{z}\right) f_Z(z) \frac{1}{z} \mathrm{d}z = \frac{1}{2}\mathcal{K}_0(\sqrt{a}), a > 0,
\end{equation}
where $\mathcal{K}_{\upsilon}(\cdot)$ is the $v$th order modified Bessel function of the second kind\setcitestyle{open=[,close=]}\cite{Gradshteyn2007Table}, and we use the integral in\setcitestyle{open={},close={}}[\cite{Gradshteyn2007Table}, Eq. (3.4719)]. By further invoking the asymptotic result\setcitestyle{open=[,close=]}\cite{Alan1989Eigenvalues}
\begin{equation} \label{lambda}
\frac{1}{M} \sigma_{1m}^2 \rightarrow (1 + \sqrt{N/M})^2
\end{equation}
for large $M$ and $N$ but a constant $N/M$, we rewrite (\ref{R1upper2}) as
\begin{equation}
\begin{split}
R_1 &\leq \frac{1}{2} \log_2\left( \frac{P \alpha^2 \sigma_{BC}^2 \sigma_{RB}^2}{\beta\sigma^2} \right) \int_0^{\infty} \mathcal{K}_0(\sqrt{a}) \mathrm{d}a + \frac{\log_2e}{2} \int_0^{\infty} \ln\left(1+\frac{\beta}{4}a\right) \mathcal{K}_0(\sqrt{a}) \mathrm{d}a \\
&\overset{(f)}= \log_2\left( \frac{P \alpha^2 \sigma_{BC}^2 \sigma_{RB}^2}{\beta\sigma^2} \right) \int_0^{\infty}a\mathcal{K}_0(a) \mathrm{d}a + \frac{\log_2e}{4} \int_0^{\infty} G\begin{smallmatrix}1\ 2\\2\ 2\end{smallmatrix} \left[\frac{\beta}{4} a\left| \begin{smallmatrix}1&1\\1&0\end{smallmatrix} \right.\right] G\begin{smallmatrix}2\ 0\\0\ 2\end{smallmatrix} \left[\frac{a}{4}\left| \begin{smallmatrix}\cdot&\cdot\\0&0\end{smallmatrix} \right.\right] \mathrm{d}a\\
&\overset{(g)}= \log_2\left( \frac{P \alpha^2  \sigma_{BC}^2 \sigma_{RB}^2}{\beta \sigma^2} \right) + \log_2e\frac{1}{\beta} G\begin{smallmatrix}4\ 1\\2\ 4\end{smallmatrix} \left[\frac{1}{\beta}\left| \begin{smallmatrix}-1&0\\0&0&-1&-1\end{smallmatrix} \right.\right] \\
\end{split}
\end{equation}
where $(f)$ uses the variable substitution and\setcitestyle{open={},close={}}[\cite{Adamchik1990The}, Eqs. (11)\&(14)] and $(g)$ results from the integrals in [\cite{Gradshteyn2007Table}, Eq. (6.56116)] and [\cite{Adamchik1990The}, Eq. (21)].

Proving (\ref{R2equ}) is analogous. By substituting (\ref{uandvs}), (\ref{SNR2}) and $\mathbf{v}_c = \widetilde{\mathbf{h}}_{BC}$ into (\ref{R2}), we have
\begin{equation}
\begin{split}
R_2 &= \frac{1}{K} \mathbb{E}_h \left[ \log_2 \left( 1 + \frac{KP\alpha^2}{\sigma^2} \Vert \mathbf{h}_{BC} \Vert^2 \vert \mathbf{h}_{RB}^H \mathbf{v}_{1m} \vert^2 \right) \right] \\
&\overset{(h)}= \frac{1}{K} \int_{0}^{\infty} \log_2 \left( 1 + \frac{\gamma}{4} b \right) \times \left( \int_0^{\infty} f_Z\left(\frac{b}{y}\right) f_Y(y) \frac{1}{y} \mathrm{d}y \right) \mathrm{d}b \\
&\overset{(i)} = \frac{1}{ 2^N K\Gamma(N)} \int_0^{\infty} \log_2\left(1 + \frac{\gamma}{4} b\right) b^{\frac{N-1}{2}} \mathcal{K}_{N-1}(\sqrt{b})\mathrm{d}b \\
&\overset{(j)} = \frac{\log_2e}{ 2^{N +1}K\Gamma(N)} \int_0^{\infty} b^{\frac{N-1}{2}} G\begin{smallmatrix}1\ 2\\2\ 2\end{smallmatrix} \left[\frac{\gamma}{4}b\left| \begin{smallmatrix}1&1\\1&0\end{smallmatrix} \right.\right] G\begin{smallmatrix}2\ 0\\0\ 2\end{smallmatrix} \left[\frac{b}{4}\left| \begin{smallmatrix}\cdot&\cdot\\\frac{N-1}{2}&-\frac{N-1}{2}\end{smallmatrix} \right.\right] \mathrm{d}b \\
&\overset{(l)}=\frac{\log_2e}{K\Gamma(N)} \gamma^{-\frac{N+1}{2}} G\begin{smallmatrix}4\ 1\\2\ 4\end{smallmatrix} \left[\gamma^{-1}\left| \begin{smallmatrix}-\frac{N+1}{2}&-\frac{N-1}{2}\\\frac{N-1}{2}&-\frac{N-1}{2}&-\frac{N+1}{2}&-\frac{N+1}{2}\end{smallmatrix} \right.\right],
\end{split}
\end{equation}\setcitestyle{open={},close={}}where $(h)$ follows from the fact that $\Vert \mathbf{h}_{BC} \Vert^2$ and $\vert \mathbf{h}_{RB}^H \mathbf{v}_{1m} \vert^2$ are independent, $(i)$ uses the integral in [\cite{Gradshteyn2007Table}, Eq. (3.4719)], and $(j)$ and $(l)$ use [\cite{Adamchik1990The}, Eqs. (11), (14)\&(21)].
\end{proof}

\subsection{Rate Bound Characterization}
The expressions in \emph{Lemma 1} with the Meijer's G-function are still too complicated to obtain insights. We further present tight bounds of the rates in the following. The tightness of the bounds will be evaluated by numerical exemplifications in the next section.

\begin{mytheorem} \label{Proposition}
The ergodic rates can be upper bounded by the closed-form expressions
\begin{equation} \label{R1upper3}
R_1 \leq \log_2 \left( 1 + \frac{P}{\sigma^2} \left( \left(\sqrt{M} + \sqrt{N}\right)^2 + \alpha^2 \sigma_{BC}^2 \sigma_{RB}^2 \right) \right) \triangleq \bar{R}_1,
\end{equation}
\begin{equation} \label{R2upper}
R_2 \leq \frac{1}{K} \log_2 \left( 1 + \frac{P\sigma_{BC}^2 \sigma_{RB}^2}{\sigma^2} K N\alpha^2 \right) \triangleq \bar{R}_2.
\end{equation}
\end{mytheorem}

\begin{proof}
From $R_1$ in (\ref{R1upper2}), we further apply Jensen's Inequality with $h$ in the expectation, it gives
\begin{equation} \label{R1upper4}
\begin{split}
R_1 &\leq \log_2 \bigg( 1 + \frac{P}{\sigma^2} \Big(\mathbb{E}_h \left[\sigma_{1m}^2\right] + \alpha^2 \mathbb{E}_h \left[ \vert\mathbf{u}_{1m}^H \mathbf{h}_{BC}\vert^2 \right] \mathbb{E}_h \left[ \vert \mathbf{h}_{RB}^H \mathbf{v}_{1m}\vert^2 \right] \Big) \bigg),
\end{split}
\end{equation}
where we also use the independence between the variables. Given the distributions of the variables as in the proof of \emph{Lemma 1}, we have
\begin{align}
&\mathbb{E}_h \left[\sigma_{1m}^2\right]\rightarrow \left(\sqrt{M} + \sqrt{N}\right)^2, \\
&\mathbb{E}_h \left[ \vert\mathbf{u}_{1m}^H \mathbf{h}_{BC}\vert^2 \right] = \sigma_{BC}^2, \ \mathbb{E}_h \left[ \vert \mathbf{h}_{RB}^H \mathbf{v}_{1m}\vert^2 \right] = \sigma_{RB}^2.
\end{align}
Substituting those results in (\ref{R1upper4}) completes the proof of (\ref{R1upper3}). The proof of (\ref{R2upper}) is similar and it is omitted here.
\end{proof}

\begin{myremark} \label{remark1}
It is concluded from (\ref{R1upper3}) that the transmitted power of an $M$-antenna RF source in the cooperative ambient backscatter system with $N$ receive antennas can be approximately reduced by the proportion of $1/(\sqrt{M} + \sqrt{N})^2$ for a nonvanishing rate.
\end{myremark}

\begin{myremark}
The ergodic rate of the conventional data link in (\ref{R1upper3}) increases slightly with the power of the backscattered signal because the BTx also unintentionally serves as a relay in the network.
\end{myremark}

\begin{myremark} \label{remark3}
As shown in (\ref{R2upper}), the ergodic rate of the backscatter link increases with the number of receive antennas and decreases with the transmission period. In particular, $\bar{R}_2$ scales like $\frac{1}{K} \log_2KN$. It implies that more receive antennas can effectively compensate for a larger symbol period to achieve certain transmission requirement.
\end{myremark}

\begin{myremark}
It can also be seen directly from (\ref{R2upper}) that the ergodic rate of the backscatter link depends logarithmically on $\alpha^2$. The BTx with greater backscattered power has larger transmission rate.
\end{myremark}

\section{Simulation Results}
In this section, we provide numerical results to validate the theoretical derivations. We set $K = 15$ and the reflection coefficient of BTx is $\alpha = 0.5$. All simulation results are averaged over 1000 channel realizations.

\begin{figure*}
	\begin{minipage}[t]{0.5\linewidth}
		\includegraphics[width=3.2in,trim=25 5 35 20,clip]{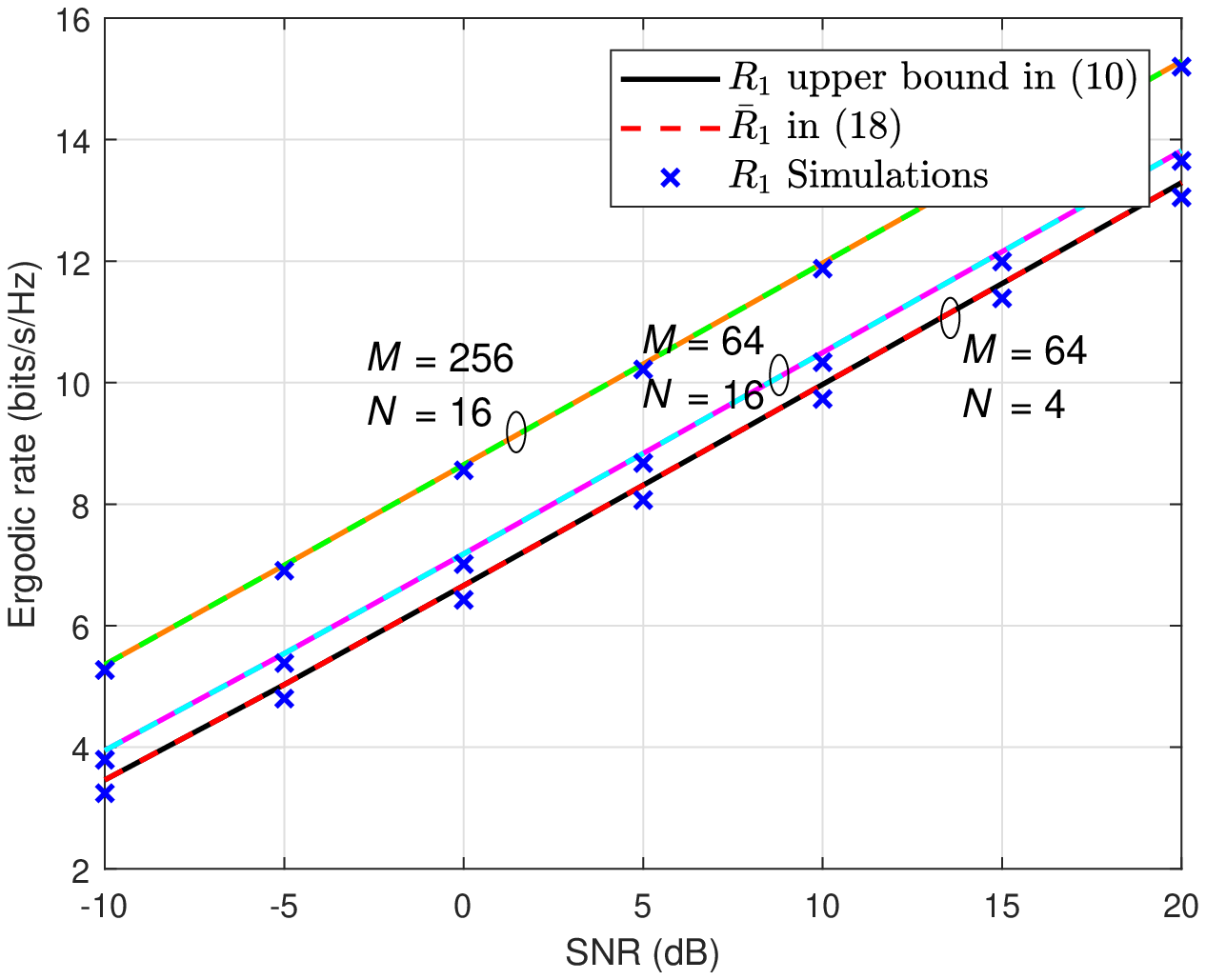}
		\caption{Ergodic rate $R_1$ versus SNR.}
		\label{Fig:R1comp}
	\end{minipage}
	\begin{minipage}[t]{0.5\linewidth}
		\includegraphics[width=3.2in,trim=25 5 35 20,clip]{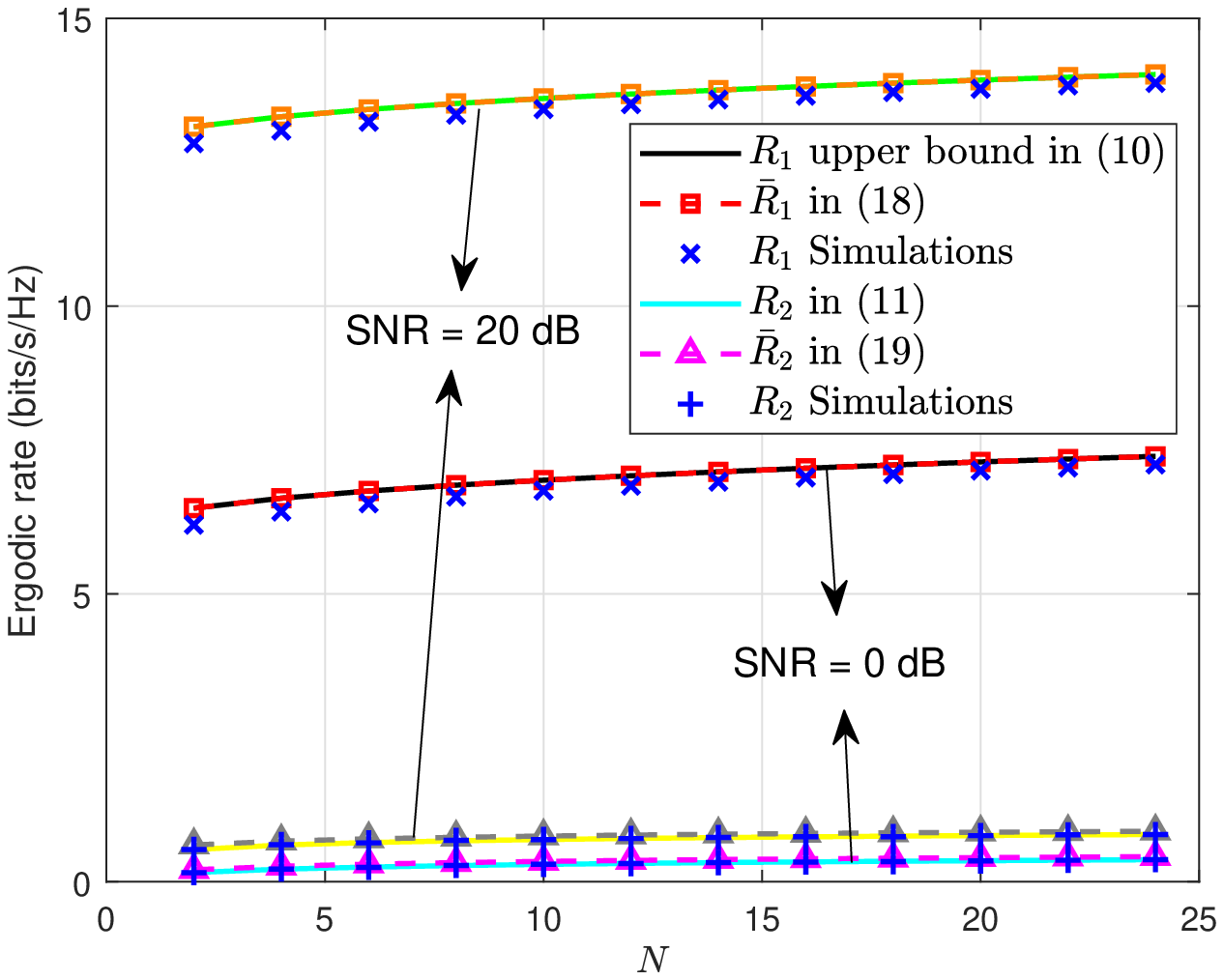}
		\caption{Ergodic rates $R_1$ and $R_2$ versus $N$.}
		\label{Fig:R2comp}
	\end{minipage}
\end{figure*}

Fig. \ref{Fig:R1comp} shows the ergodic rate of the primary link. Both bounds in (\ref{R1upper1}) and (\ref{R1upper3}) are presented for comparison. It shows that though large antenna number is assumed, the derived bound appears to be fairly tight even with small number of antennas, e.g., $N = 4$. The gap diminishes as the numbers of antennas increase. In Fig. \ref{Fig:R1comp}, when the number of RF source antennas and that of receive antennas increase four times simultaneously, the ergodic rate increases by 2 bps/Hz, which is consistent with \emph{Remark \ref{remark1}}.

\begin{figure}
  \setlength{\abovecaptionskip}{0.cm}
  \setlength{\belowcaptionskip}{-0.cm}
  \centering
  \includegraphics[width=3.2in,trim=25 5 35 20,clip]{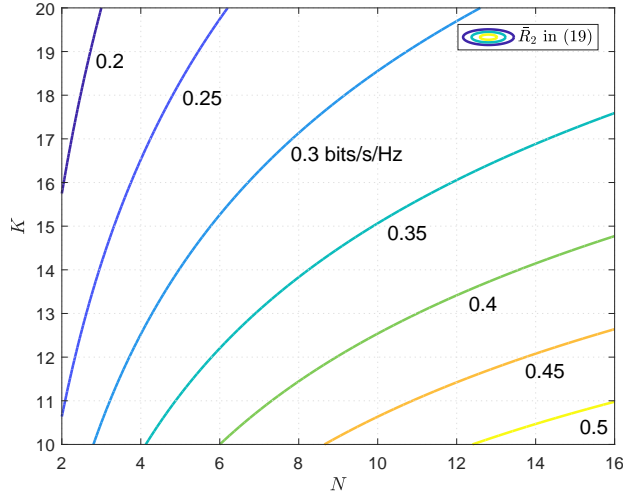}
  \caption{Ergodic rate $\bar{R}_2$ as a function of $N$ and $K$.}
  \label{Fig:Response1}
\end{figure}

We further assume that $M = 64$. In Fig. \ref{Fig:R2comp}, we plot the ergodic rates of both links obtained in \emph{Lemma 1}, \emph{Theorem 1} and by simulations. For $R_2$, the derived bound matches the simulation results well even with a small value of $N = 2$. As the number of receive antennas gets larger, the two ergodic rates increase with the respective orders as predicted in \emph{Remark \ref{remark1}} and \emph{Remark \ref{remark3}}. In Fig. \ref{Fig:Response1}, we plot the ergodic rate of the backscatter link as a function of $N$ and $K$. For each contour of $\bar{R}_2$, when the number of receive antennas increases, the ratio of transmission periods grows too, which corresponds to \emph{Remark \ref{remark3}}.

\section{Conclusion}
In this paper, we demonstrate the rate bounds of a cooperative ambient backscatter communication system. The transmit power of the RF source can be approximately reduced by the proportion of $1/(\sqrt{M} + \sqrt{N})^2$ with increasing numbers of antennas. The ergodic rate of the backscatter link asymptotically behaves like $\frac{1}{K} \log_2KN$. The two parameters, i.e., the transmission period and the number of receive antennas, can be cooperatively adjusted to achieve the desired performance. For our future work, it is interesting to obtain a tractable lower bound of the system rate for analysis.

\end{document}